\documentclass[a4paper,11pt]{article}
\usepackage{fullpage}
\usepackage{amsmath,amsthm,amssymb,amsfonts,mathrsfs,amsthm,centernot}
\usepackage{dsfont}
\usepackage{graphicx,hyperref}
\usepackage{enumitem}
\usepackage[perpage,symbol*]{footmisc}

\def\Tr{\operatorname{Tr}}
\def\>{\rangle}
\def\<{\langle}
\def\N#1{\left|\!\left|{#1}\right|\!\right|}
\def\id{\mathsf{id}}
\def\mE{\mathcal{E}}

\def\mN{\mathcal{N}}
\def\mL{\mathcal{L}}

\def\mD{\mathcal{D}}
\def\mT{\mathcal{T}}
\def\sH{\mathcal{H}}

\def\openone{\mathds{1}}
\newcommand{\set}[1]{\mathcal{#1}}
\newcommand{\op}[1]{\mathsf{#1}}

\newcommand{\pg}{P_{\operatorname{guess}}}

\newcommand{\defeq}{\stackrel{\textup{\tiny def}}{=}}

\newcommand{\MN}[1]{\left|\!\left|\!\left|#1\right|\!\right|\!\right|}

\newcommand{\Hmin}{H_{\operatorname{min}}}

\renewcommand{\qedsymbol}{\nobreak \ifvmode \relax \else
  \ifdim \lastskip<1.5em \hskip-\lastskip \hskip1.5em plus0em
  minus0.5em \fi \nobreak \vrule height0.75em width0.5em
  depth0.25em\fi}

\renewcommand{\ge}{\geqslant}
\renewcommand{\le}{\leqslant}

\newtheorem{corollary}{Corollary}
\newtheorem{lemma}{Lemma}
\newtheorem{proposition}{Proposition}

\newtheorem{definition}{Definition}

\theoremstyle{remark}
\newtheorem{remark}{Remark}

\theoremstyle{definition}

\newcommand{\bea}{\begin{eqnarray}}
\newcommand{\eea}{\end{eqnarray}}
\newcommand{\be}{\begin{equation}}
\newcommand{\ee}{\end{equation}}

\begin{document}


\title{Degradable channels, less noisy channels,\\and quantum statistical morphisms: an equivalence relation}

\author{{\small \sc Francesco Buscemi}\footnote{buscemi@is.nagoya-u.ac.jp}\\
    \footnotesize{\em Dept. of Computer Science and Mathematical Informatics}\\
    \footnotesize{\em Nagoya University}}

\date{\today}

\maketitle

\begin{abstract}
	Two partial orderings among communication channels, namely, `being degradable into' and `being less noisy than,' are reconsidered in the light of recent results about statistical comparisons of quantum channels. Though our analysis covers at once both classical and quantum channels, we also provide a separate treatment of classical noisy channels, and show how, in this case, an alternative self-contained proof can be constructed, with its own particular merits with respect to the general result.
\end{abstract}


\section{Introduction}

Given two channels, it is natural to ask which one is `better.' As one soon realizes, ordering channels according to their capacity would be too limited in scope, as there may be other measures of `goodness' that are more relevant than capacity for the task at hand. This is even more so for \textit{quantum} channels, for which many inequivalent capacities exist~\cite{wilde}. Indeed, it does not require too much imagination to come up with uncountably many such `comparisons,' and each comparison will only induce a \textit{partial}, rathen than a \textit{total}, ordering between channels. This fact should not come as a worrying surprise: channels are highly dimensional objects and any total ordering can only be extremely coarse---often too coarse to be of any use in practice. Nonetheless, it is true that some comparisons are more natural, more compelling, or just mathematically simpler than others, so that they received more attention in the literature. This paper actually deals with two much studied comparisons, namely, the partial orderings `being degradable into' and `being less noisy than,' introduced in Definitions~\ref{def:degradable} and~\ref{def:less-noisy} below (for a compendium of many comparisons among discrete noisy channels see Ref.~\cite{cohen_comparisons_1998}).

The goal of this work is to exhibit a connection between degradable channels and less noisy channels, beyond the obvious one `degradable implies less noisy.' More explicitly, we show how a formal (but not substantial) modification in the definition of less noisy channels is sufficient to make the two orderings equivalent. Our result is proved in a general scenario, where channels are modeled as CPTP maps between operator algebras, thus covering quantum channels, classical channels (when input and output are commutative algebras), but also hybrid classical-to-quantum and quantum-to-classical channels.

Central to our approach is the notion of \textit{quantum statistical morphisms}, namely, linear maps between operator algebras that generalize in a statistical sense the idea of `post-processing' or `coarse-graining' (see Definition~\ref{def:morphisms}). The use of statistical morphisms allows us to prove our results under very mild assumptions, so that the quantum and classical cases are recovered as special cases of a single unifying framework. Such an top-down approach, besides being mathematically simpler, it has the merit to clearly separate statistics from physics: indeed, it can be immediately applied to general probabilistic theories, as it does not rely on any particular feature of quantum theory like, for example, complete positivity.

The paper is organized as follows. In Section~\ref{sec:notation} we introduce the notation and some basic definitions. In Section~\ref{sec:general} we introduce the concept of statistical morphisms and prove the fundamental equivalence relation. In Section~\ref{sec:classical} we specialize to the case of semiclassical channels. The case of discrete noisy classical channels is treated separately in Subsection~\ref{sec:fully-class}, in a way that does not rely on any knowledge of the quantum case and that allows the treatment of the approximate case, studied in Appendix~\ref{sec:app-approximate}. In Section~\ref{sec:quantum} we consider the case of fully quantum channels. Section~\ref{sec:concls} concludes the paper with some comments and possible future developments.

\section{Notations and definitions}\label{sec:notation}

In what follows, all sets are finite and Hilbert spaces are finite-dimensional.
\begin{itemize}
	\item Sets are denoted by $\set{X}=\{x:x\in\set{X}\}$, $\set{Y}=\{y:y\in\set{Y}\}$, etc.
	\item A probability distribution over $\set{X}$ is a function $p:\set{X}\to[0,1]$ such that $\sum_xp(x)=1$.
	\item The \textbf{set of all probability distributions} over $\set{X}$ is denoted by $\op{P}(\set{X})$.
	\item Random variables are labeled by upper case letters $X$, $Y$, etc, with ranges $\set{X}=\{x\}$, $\set{Y}=\{y\}$, etc.
	\item Discrete noisy channels are identified with the associated conditional probability distributions.
	\item Quantum systems are labeled by upper case letters $Q$, $R$, etc, and the associated Hilbert spaces are denoted by $\sH_Q$, $\sH_R$, etc. Dimensions are denoted as $d_Q\defeq\dim\sH_Q$ etc.
	\item The \textbf{set of linear operators} acting on a Hilbert space $\sH$ is denoted by $\op{L}(\sH)$. The \textbf{identity operator} is denoted by $\openone$.
	\item States of $Q$ are represented by \textbf{density operators}, i.e., operators $\rho\in\op{L}(\sH)$ such that $\rho\ge0$ and $\Tr[\rho]=1$.
	\item The \textbf{set of density operators} acting on a Hilbert space $\sH$ is denoted by $\op{S}(\sH)$.
	\item A positive-operator valued measure (\textbf{POVM}) is a function $P:\set{X}\to\op{L}(\sH)$ such that $P(x)\ge0$ and $\sum_xP(x)=\openone$. For the sake of readability, we will often write the argument $x$ as a superscript, i.e., $P^x$ rather than $P(x)$.
	\item The \textbf{set of POVMs} from $\set{X}$ to $\op{L}(\sH)$ is denoted by $\op{M}(\set{X},\sH)$.
	\item \textbf{Quantum channels} are completely positive trace-preserving (CPTP) linear maps $\mN:\op{L}(\sH_Q)\to\op{L}(\sH_{R})$. The \textbf{range} of a channel $\mN$ is defined as the image of $\op{L}(\sH_Q)$ under the action of $\mN$, namely, the set $\{\mN(X):X\in\op{L}(\sH_Q) \}$. The \textbf{identity map} is denoted by $\id$.
	\item The \textbf{set of quantum channels} from $\op{L}(\sH_Q)$ to $\op{L}(\sH_R)$ is denoted by $\op{C}(\sH_Q,\sH_R)$.
	\item Given a linear map $\mL:\op{L}(\sH_Q)\to\op{L}(\sH_R)$, its \textbf{trace dual} is the linear map $\mL^*:\op{L}(\sH_R)\to\op{L}(\sH_Q)$ defined by the relation:
	\[
	\Tr[\mL^*(X_R)\ Y_Q]\defeq\Tr[X_R\ \mL(Y_Q)],
	\]
	for all $Y_Q\in\op{L}(\sH_Q)$ and all $X_R\in\op{L}(\sH_R)$. Then, $\mN$ is trace-preserving if and only if $\mN^*$ is \textbf{unit-preserving}, i.e., $\mN^*(\openone_R)=\openone_Q$.  Moreover, we say that a linear map $\mL$ is \textbf{Hermitian} if and only if, for any $X=X^\dag$, $\mL(X)=\mL(X)^\dag$.
	\item A \textbf{classical-to-quantum (cq) channel} is a function $\mE:\set{X}\to\op{S}(\sH)$. We will usually denote the density operators $\mE(x)$ by $\rho^x$, $\sigma^x$, etc. Equivalently, a cq-channel $\mE$ will be denoted as a family of density operators $\mE=\{\rho^x:x\in\set{X} \}$.
	\item A \textbf{classical-quantum (cq) state} is a bipartite density operator describing a quantum system $Q$ correlated with a random variable $X$. Since random variables can be seen as commuting density operators, we will represent cq-states as, e.g., $\rho_{XQ}=\sum_{x\in\set{X}}p(x)|x\>\<x|_X\otimes\rho^x_Q$, where the unit vectors $\{|x\>:x\in\set{X} \}$ are all orthogonal. 
	\item For a given bipartite density operator $\rho_{RQ}\in\op{S}(\sH_R\otimes\sH_Q)$, its \textbf{conditional min-entropy} $\Hmin(R|Q)_\rho$ is defined as
	\[
	\Hmin(R|Q)_\rho\defeq -\inf_{\sigma_Q\in\op{S}(\sH_Q)}\inf\{\lambda\in\mathbb{R}:\rho_{RQ}\le2^\lambda \openone_R\otimes\sigma_Q  \}.
	\]
	We will use in particular the fact that~\cite{konig2009operational}
	\[
	2^{-\Hmin(R|Q)_\rho}=d_R\max_{\mN\in\op{C}(\sH_Q,\sH_{R'})}F^2((\id_R\otimes\mN)\rho_{RQ},\Phi^+_{RR'}),
	\]
	where $\sH_{R'}\cong\sH_R$, $F^2(\rho,\sigma)\defeq \N{\sqrt{\rho}\sqrt{\sigma}}_1^2$, and $\Phi^+_{RR'}\defeq d_R^{-1}\sum_{i,j=1}^{d_R}|i_R\>|i_{R'}\>\<j_R|\<j_{R'}|$, for some orthonormal basis $\{|i\> \}$ of $\sH_R$.
	
	In the case of a cq-state $\rho_{XQ}=\sum_{x\in\set{X}}p(x)|x\>\<x|_X\otimes\rho^x_Q$, the above formula becomes equivalent to
	\[
	2^{-\Hmin(X|Q)_\rho}=\max_{P\in\op{M}(\set{X},\sH_Q)}\sum_{x\in\set{X}}p(x)\Tr[\rho^x_Q\ P^x_Q]\defeq \pg(X|Q)_\rho,
	\]
	namely, the expected \textbf{guessing probability}, namely, the probability of correctly guessing the value of $X$ having access only to the quantum system $Q$. 
\end{itemize}

\subsection{`Degradable'  and `less noisy' channels}

In the classical case, the following definitions can be found in Refs.~\cite{bergmans,Korner1977,ElGamal1977,csiszar-korner}.

\begin{definition}[Degradable channels]\label{def:degradable}
	Given two discrete channels $p(y|x)$ and $p'(z|x)$,
	$p$ is said to be {\em degradable} into $p'$ whenever there exists a discrete channel $q(z|y)$ such that
	\[
	p'(z|x)=\sum_{y\in\mathcal{Y}}q(z|y)p(y|x).
	\]
\end{definition}

\begin{definition}[Less noisy channels]\label{def:less-noisy}
	Given two discrete channels $p(y|x)$ and $p'(z|x)$,
	$p$ is said to be {\em less noisy} than $p'$ whenever, for any discrete random variable $U$, any probability distribution $q(u)$, and any channel $q(x|u)$, the joint input-output probability distributions $q(u)q(x|u)p(y|x)$ and $q(u)q(x|u)p'(z|x)$ satisfy
	\[
	H(U|Y)\le H(U|Z).
	\]
\end{definition}

If $p_1$ is degradable into $p_2$, then $p_1$ is less noisy than $p_2$: the proof is a simple consequence of the data-processing inequality. Counterexamples are known for the converse~\cite{Korner1977}, namely, one channel can be less noisy than another without being degradable. A consequence of the results presented here is that it is sufficient to replace $H$ with $\Hmin$ in  Definition~\ref{def:less-noisy} in order to make the ordering `less noisy' (now defined with respect to $\Hmin$) equivalent to the ordering `degradable.' This fact is formalized in Corollary~\ref{coro:classical} below. (The reader interested only in the classical case can directly skip to Subsection~\ref{sec:fully-class}: there,  Corollary~\ref{coro:classical} is provided an independent, self-contained proof, which does not rely on any idea developed for the general non-commutative case. Moreover, such a proof allows the treatment of the approximate case, which is studied in Appendix~\ref{sec:app-approximate}.)

In fact, our analysis will not be limited to the case of classical noisy channels, but will include some results valid for quantum channels too. We hence generalize Definitions~\ref{def:degradable} and~\ref{def:less-noisy} to the quantum case as follows (but compare with~\cite{watanabe}):

\begin{definition}[Degradable quantum channels]\label{def:degradable-quantum}
	Given two CPTP maps $\mN:\op{L}(\sH_Q)\to\op{L}(\sH_R)$ and $\mN':\op{L}(\sH_Q)\to\op{L}(\sH_S)$, $\mN$ is said to be {\em degradable} into $\mN'$ whenever there exists a CPTP map $\mT:\op{L}(\sH_R)\to\op{L}(\sH_S)$ such that
	\[
	\mN'=\mT\circ\mN.
	\] 
\end{definition}

\begin{definition}[Less noisy quantum channels]\label{def:less-noisy-quantum}
	Given two CPTP maps $\mN:\op{L}(\sH_Q)\to\op{L}(\sH_R)$ and $\mN':\op{L}(\sH_Q)\to\op{L}(\sH_S)$, $\mN$ is said to be {\em less noisy} than $\mN'$ whenever, for any discrete random variable $U$, any probability distribution $q(u)$, and any cq-channel $\mE=\{\rho^u_Q:u\in\mathcal{U} \}$, the corresponding input-output cq-states \[\sigma_{UR}\defeq \sum_uq(u)|u\>\<u|_U\otimes\mN_Q(\rho^u_Q)\qquad \textrm{and} \qquad \tau_{US}\defeq \sum_uq(u)|u\>\<u|_U\otimes\mN'_Q(\rho^u_{Q})\] satisfy
	\[
	H(U|R)_\sigma\le H(U|S)_\tau.
	\] 
\end{definition}

This paper studies the relations between the notions of degradable channels and less noisy channels, both in classical and quantum information theory. In what follows we will show, in particular, how Definitions~\ref{def:less-noisy} and~\ref{def:less-noisy-quantum} can be formally modified so that the two partial orderings become equivalent. The results presented here are based on recent formulations of the Blackwell-Sherman-Stein Theorem~\cite{blackwell_equivalent_1953,torgersen_comparison_1991,liese_statistical_2008} for quantum systems~\cite{buscemi2005clean,shmaya,chefles,buscemi2012comparison,buscemi_game-theoretic_2014,buscemi_equivalence_2014}.

\section{Statistical morphisms and a fundamental equivalence relation}\label{sec:general}

We begin with a definition, generalizing that given in~\cite{buscemi2012comparison}:

\begin{definition}[Quantum statistical morphisms]\label{def:morphisms}
	Given a CPTP map $\mN:\op{L}(\sH_Q)\to\op{L}(\sH_R)$, a \emph{statistical morphism} of $\mN$ is a linear map $\mL:\op{L}(\sH_R)\to\op{L}(\sH_S)$ such that, for any finite outcome set $\set{X}$ and any POVM $\{\bar P^x_S:x\in\set{X} \}$, there exists another POVM $\{P^x_R:x\in\set{X} \}$ such that
\begin{equation}\label{eq:stat-cond-morph}
	\Tr[(\mL\circ\mN)(\rho_Q)\ \bar P^x_S]=\Tr[\mN(\rho_Q)\ P^x_R],\qquad\forall x\in\set{X}, \forall\rho_Q\in\op{S}(\sH_Q).
\end{equation}
\end{definition}

\begin{remark}\label{rem:morphisms}
	Clearly, a positive trace-preserving linear map is always a well-defined statistical morphisms, \textit{for any} channel. However, a map can be a statistical morphism \textit{of some channel} without being positive and trace-preserving---in fact, statistical morphisms cannot even \textit{be extended}, in general, to positive trace-preserving maps, as shown in Ref.~\cite{matsumoto_example_2014} by means of an explicit counterexample~\footnote{About this problem see also Ref.~\cite{reeb-extensions}.}. We can only say that, if $\mL$ is a statistical morphism of $\mN$, then $\mL$ is positive and trace-preserving on the range of $\mN$, namely, $\Tr[(\mL\circ\mN)(X_Q)]=\Tr[\mN(X_Q)]$, for all $X_Q\in\op{L}(\sH_Q)$, and, whenever $\mN(X_Q)\ge 0$, $\Tr[(\mL\circ\mN)(X_Q)\ \bar P_S]\ge 0$, for all $\bar P_S\ge 0$. The question then arises: is any linear map $\mL$, which is positive and trace-preserving on the range of a channel $\mN$, a statistical morphism of $\mN$? Again, the answer is no. This is because, in order to guarantee that $\mL$ is positive and trace-preserving on the range of $\mN$, it would be enough to have Eq.~(\ref{eq:stat-cond-morph}) to hold only for binary POVMs (i.e., \textit{effects}) $\{\bar P,\openone-\bar P \}$, but this condition is known to be strictly weaker than that required in Definition~\ref{def:morphisms}, which must hold for any finite $\set{X}$~\cite{jencova_comparison_2012}. The situation can thus be summarized as follows:
	\begin{equation*}
	\textrm{PTP everywhere}\underset{\centernot\impliedby}{\implies}\textrm{stat. morph. of $\mN$}\underset{\centernot\impliedby}{\implies}\textrm{PTP on $\operatorname{range}(\mN)$}.
	\end{equation*}
\end{remark}

\begin{remark}
	In what follows, when we say ``trace-preserving statistical morphism'' we mean a trace-preserving (everywhere) linear map that is, in particular, a statistical morphism (for some channel).
\end{remark}

We are now ready to state a fundamental equivalence relation.

\begin{proposition}\label{prop:morphism}
	Given two CPTP maps $\mN:\op{L}(\sH_Q)\to\op{L}(\sH_R)$ and $\mN':\op{L}(\sH_Q)\to\op{L}(\sH_S)$, the following are equivalent:
	\begin{enumerate}[label=\roman*)]
		\item for any discrete random variable $U$, any probability distribution $q(u)$, and any cq-channel $\mE=\{\rho^u_Q:u\in\set{U} \}$, the corresponding input-output cq-states \[\sigma_{UR}=\sum_{u\in\set{U}}q(u)|u\>\<u|_U\otimes\mN(\rho^u_Q)\qquad \textrm{and} \qquad \tau_{US}=\sum_{u\in\set{U}}q(u)|u\>\<u|_U\otimes\mN'(\rho^u_Q)\] satisfy
		\[
	    \Hmin(U|R)_\sigma\le		\Hmin(U|S)_\tau,
		\]
		namely, $\pg(U|R)_\sigma\ge		\pg(U|S)_\tau$;
		\item there exists a Hermitian trace-preserving statistical morphism $\mL:\op{L}(\sH_R)\to\op{L}(\sH_S)$ of $\mN$ such that
		\[
		\mN'=\mL\circ\mN.
		\]
	\end{enumerate}
\end{proposition}

Notice that point~(i) in Proposition~\ref{prop:morphism} looks exactly as the definition of less noisy channels (Def.~\ref{def:less-noisy-quantum}), the only difference being the use of $\Hmin$ in the place of $H$.

\begin{proof}	
	If point~(ii) holds, then,
	\begin{equation*}
	\begin{split}
	\pg(U|S)_\tau&=\max_{\bar P\in\op{M}(\set{U},\sH_S)}\sum_{u\in\set{U}}q(u)\Tr[\mN'(\rho^u_Q)\ \bar P^u_S]\\
	&=\max_{\bar P\in\op{M}(\set{U},\sH_S)}\sum_{u\in\set{U}}p(u)\Tr[(\mL\circ\mN)(\rho^u_Q)\ \bar P^u_S]\\
	&\le\max_{P\in\op{M}(\set{U},\sH_R)}\sum_{u\in\set{U}}p(u)\Tr[\mN(\rho^u_Q)\ P^u_R]\\
	&=\pg(U|R)_\sigma,
	\end{split}
	\end{equation*}
	where the inequality is a consequence of Definition~\ref{def:morphisms} above.
	
	Conversely, suppose that point~(i) holds. As already shown in Refs.~\cite{buscemi2012comparison,buscemi_game-theoretic_2014,buscemi_equivalence_2014}, this implies that, for any POVM $\{\bar P^x_S:x\in\set{X} \}$ on $\sH_S$, there exists a POVM $\{P^x_R:x\in\set{X} \}$ on $\sH_R$ such that
	\begin{equation}\label{eq:stat-eqiv}
	\Tr[\mN'(\rho_Q)\ \bar P^x_S]=\Tr[\mN(\rho_Q)\ P^x_R],
	\end{equation}
	for all $x\in\set{X}$ and all $\rho_Q\in\op{S}(\sH_Q)$. Let us choose $\{\bar P^x_S:x\in\set{X} \}\in\op{M}(\set{X},\sH_S)$ to be an informationally complete POVM, namely, such that $\operatorname{span}\{\bar P^x_S:x\in\set{X}\}=\op{L}(\sH_S)$. Let $\{P^x_R:x\in\set{X} \}\in\op{M}(\set{X},\sH_R)$ be the corresponding POVM, satisfying Eq.(\ref{eq:stat-eqiv}), and define a linear map $\mL^*:\op{L}(\sH_S)\to\op{L}(\sH_R)$ by
	\[
	\mL^*(\bar P^x_S)\defeq P^x_R,\qquad x\in\set{X}.
	\]
	Such a map is uniquely defined, since $\{\bar P^x_S:x\in\set{X}\}$ is a basis for $\op{L}(\sH_S)$, and it is unit-preserving by construction, implying that its trace dual $\mL:\op{L}(\sH_R)\to\op{L}(\sH_S)$ is trace-preserving. In order to prove that $\mL$ is also Hermitian, let $\{\Theta^x_S:x\in\set{X} \}$ be the set of Hermitian operators in $\op{L}(\sH_S)$ such that
	\[
	X_S=\sum_{x\in\set{X}}\Tr[X_S\ \bar P^x_S]\ \Theta^x_S,
	\]
	for all $X_S\in\op{L}(\sH_S)$. This implies that, for any $Y_R=Y_R^\dag$ in $\op{L}(\sH_R)$,
	\[
	\begin{split}
	\mL(Y_R)=&\sum_{x\in\set{X}}\Tr[\mL(Y_R)\ \bar P^x_S]\Theta^x_S\\
	=&\sum_{x\in\set{X}}\Tr[Y_R\ \mL^*(\bar P^x_S)]\Theta^x_S\\
	=&\sum_{x\in\set{X}}\Tr[Y_R\ P^x_R]\Theta^x_S\\
	=&\sum_{x\in\set{X}}\lambda_x\Theta^x_S\\
	\end{split}
	\]
	with $\lambda_x\in\mathbb{R}$, i.e., $\mL(Y_R)$ is Hermitian too. Hence, $\mL$, as defined above, is a Hermitian trace-preserving linear map. We only need to show that $\mN'=\mL\circ\mN$ and that $\mL$ is well-defined statistical morphism of $\mN$.
	
	In order to show that $\mN'=\mL\circ\mN$, we notice that the condition expressed in Eq.~(\ref{eq:stat-eqiv}) can be reformulated as follows: for any $\rho_Q\in\op{S}(\sH_Q)$ and any $x\in\set{X}$,
	\begin{equation*}
	\begin{split}
	\Tr[\mN'(\rho_Q)\ \bar P^x_S]&=\Tr[\mN(\rho_Q)\ P^x_R]\\
	&=\Tr[\mN(\rho_Q)\ \mL^*(\bar P^x_S)]\\
	&=\Tr[(\mL\circ\mN)(\rho_Q)\ \bar P^x_S].
	\end{split}
	\end{equation*}
	Since $\{\bar P^x_S:x\in\set{X}\}$ in the above equation is informationally complete, we have that $\mN'(\rho_Q)=(\mL\circ\mN)(\rho_Q)$, for all $\rho_Q$, namely, $\mN'=\mL\circ\mN$. Thus we also know that the condition expressed in Eq.~(\ref{eq:stat-eqiv}) above automatically implies that $\mL$ is a well-defined statistical morphism.
\end{proof}

In other words, Proposition~\ref{prop:morphism} above states that, to replace $H$ with $\Hmin$ in Definition~\ref{def:less-noisy-quantum} is sufficient to conclude a \textit{weaker form of degradability}, in the sense that the degrading map is not a quantum channel, but a Hermitian trace-preserving statistical morphism. In what follows, we will see when one can conclude that the degrading map is in fact CPTP.

Before proceeding, however, we specialize Proposition~\ref{prop:morphism} to the case of cq-channels, which can always be seen as CPTP maps on commuting input subalgebras. We start by simplifying the definition of statistical morphism as follows (this was the original definition given in~\cite{buscemi2012comparison}):

\begin{definition}
	Given a cq-channel $\mE:\set{X}\to\op{S}(\sH_R)$ with $\mE=\{\sigma_R^x:x\in\set
		X \}$, a \emph{statistical morphism} of $\mE$ is a linear map $\mL:\op{L}(\sH_R)\to\op{L}(\sH_S)$ such that, for any $\set{Y}$ and any POVM $\{\bar P^y_S:y\in\set{Y} \}$, there exists a corresponding POVM $\{P^y_R:y\in\set{Y} \}$ such that $\Tr[\mL(\sigma^x_R)\ \bar P^y_S]=\Tr[\sigma^x_R\ P^y_R]$.
\end{definition}

\begin{remark}
Notice that, in order for $\mL$ to be a well-defined statistical morphism of $\mE$, it is not enough that $\mL(\sigma^x_R)\in\op{S}(\sH_S)$ for all $x\in\set{X}$. In particular, such a map is not, in general, positive on the whole $\operatorname{span}\{\sigma^x_R:x\in\set{X} \}$. See also Remark~\ref{rem:morphisms} for more details.
\end{remark}

\begin{proposition}\label{thm:morphism-cq}
	Given two cq-channels $\mE:\set{X}\to\op{S}(\sH_R)$ and $\mE':\set{X}\to\op{S}(\sH_S)$, with $\mE=\{\sigma_R^x:x\in\set
	X \}$ and $\mE'=\{\tau_S^x:x\in\set
	X \}$, the following are equivalent:
	\begin{enumerate}[label=\roman*)]
		\item for any discrete random variable $U$, any probability distribution $q(u)$, and any classical channel $q(x|u)$, the corresponding input-output cq-states \[\sigma_{UR}=\sum_{x\in\set{X}}\sum_{u\in\set{U}}q(u)q(x|u)|u\>\<u|_U\otimes\sigma^x_R\qquad \textrm{and} \qquad \tau_{US}=\sum_{x\in\set{X}}\sum_{u\in\set{U}}q(u)q(x|u)|u\>\<u|_U\otimes\tau^x_S\] satisfy
		\[
		\Hmin(U|R)_{\sigma}\le \Hmin(U|S)_{\tau},
		\]
		namely, $\pg(U|R)_{\sigma}\ge \pg(U|S)_{\tau}$;
		\item there exists a Hermitian trace-preserving statistical morphism of $\mE$, denoted by $\mL:\op{L}(\sH_R)\to\op{L}(\sH_S)$, such that
		\[
		\tau^x_S=\mL(\sigma^x_R),\qquad x\in\set{X}.
		\]
	\end{enumerate}
\end{proposition}

\section{First extension result: the semiclassical and classical cases}\label{sec:classical}

One sufficient condition for a statistical morphisms to be extendable to a CPTP map is that the composite map $\mL\circ\mN$ has commuting output.

\begin{lemma}\label{lemma:semiclassical}
	Let $\mL:\op{L}(\sH_R)\to\op{L}(\sH_S)$ be a statistical morphism of a given channel $\mN\in\op{C}(\sH_Q,\sH_{R})$. If \[\Big[(\mL\circ\mN)(\rho)\ ,\ (\mL\circ\mN)(\sigma)\Big]=0\] for all $\rho,\sigma\in\op{S}(\sH_Q)$, then there exists a CPTP map $\mT:\op{L}(\sH_R)\to\op{L}(\sH_S)$ such that
	\[
	\mT\circ\mN=\mL\circ\mN.
	\]
\end{lemma}

\begin{proof}
	Being $\mL$ a statistical morphism of $\mN$, we know, from Definition~\ref{def:morphisms}, that, for any POVM $\{\bar P^x_S:x\in\set{X} \}$ on $\sH_S$, there exists a POVM $\{P^x_R:x\in\set{X} \}$ on $\sH_R$ such that
	\[
	\begin{split}
	\Tr[(\mL\circ\mN)(\rho_Q)\ \bar P^x_S]&=\Tr[\mN(\rho_Q)\ \mL^*(\bar P^x_S)]\\
	&=\Tr[\mN(\rho_Q)\ P^x_R],
	\end{split}
	\]
	for all $\rho_Q\in\op{S}(\sH_Q)$. For $\set{X}=[1,d_S]$, denote by $\{|x\>:x\in\set{X} \}$ the orthonormal basis of $\sH_S$ that simultaneously diagonalize any output of $\mL\circ\mN$. (Such a basis exists, since $[(\mL\circ\mN)(\rho),(\mL\circ\mN)(\sigma)]=0$.) Then choose, in the above equation, $\bar P^x_S=|x\>\<x|_S$, and define $\mT:\op{L}(\sH_R)\to\op{L}(\sH_S)$ to be the linear map given by
	\[
	\mT(Z_R)\defeq \sum_{x\in\set{X}}|x\>\<x|_S\Tr[Z_R\ P^x_R],
	\]
	for any $Z_R\in\op{L}(\sH_R)$. By construction, $\mT$ is CPTP (indeed, it is a measure-and-prepare quantum channel).
	Moreover, for all $\rho_Q\in\op{S}(\sH_Q)$, 
	\[
	\begin{split}
	(\mT\circ\mN)(\rho_Q)&=\sum_{x\in\set{X}}|x\>\<x|_S\Tr[\mN(\rho_Q)\ P^x_R]\\
	&=\sum_{x\in\set{X}}|x\>\<x|_S\Tr[\mN(\rho_Q)\ \mL^*(|x\>\<x|_S)]\\
	&=\sum_{x\in\set{X}}|x\>\<x|_S\Tr[(\mL\circ\mN)(\rho_Q)\ |x\>\<x|_S]\\
	&=(\mL\circ\mN)(\rho_Q),
	\end{split}
	\]
	where the last identity comes from the fact that all $\mN'(\rho_Q)$ are all diagonal on the basis $\{|x\> \}$.
\end{proof}

As an immediate consequence of Lemma~\ref{lemma:semiclassical} above and Proposition~\ref{prop:morphism}, we obtain the following:

\begin{corollary}\label{coro:semiclassical}
	Let $\mN\in\op{C}(\sH_Q,\sH_{R})$ and $\mN'\in\op{C}(\sH_Q,\sH_{S})$ be two CPTP map. Let moreover $\mN'$ be such that $[\mN'(\rho),\mN'(\sigma)]=0$, for all $\rho,\sigma\in\op{S}(\sH_Q)$. Then, the following are equivalent:
	\begin{enumerate}[label=\roman*)]
		\item for any discrete random variable $U$, any probability distribution $q(u)$, and any cq-channel $\mE=\{\rho^u_Q:u\in\set{U} \}$, the corresponding input-output cq-states \[\sigma_{UR}=\sum_{u\in\set{U}}q(u)|u\>\<u|_U\otimes\mN_Q(\rho^u_Q)\qquad \textrm{and} \qquad \tau_{US}=\sum_{u\in\set{U}}q(u)|u\>\<u|_U\otimes\mN'_Q(\rho^u_Q)\] satisfy
		\[
		\Hmin(U|R)_\sigma\le		\Hmin(U|S)_\tau,
		\]
		namely, $\pg(U|R)_\sigma\ge		\pg(U|S)_\tau$;
		\item $\mN$ is degradable into $\mN'$, i.e., there exists a CPTP map $\mT:\op{L}(\sH_R)\to\op{L}(\sH_S)$ such that
		\[
		\mN'=\mT\circ\mN.
		\]
	\end{enumerate}
\end{corollary}

The above corollary can be specialized to cq-channels as follows.

\begin{corollary}
Consider two cq-channels $\mE:\set{X}\to\op{S}(\sH_R)$ and $\mE':\set{X}\to\op{S}(\sH_S)$, with $\mE=\{\sigma_R^x:x\in\set
X \}$ and $\mE'=\{\tau_S^x:x\in\set
X \}$. Assume moreover that $[\tau_S^x,\tau_S^{x'}]=0$, for all $x,x'\in\set{X}$. Then, the following are equivalent:
\begin{enumerate}[label=\roman*)]
	\item for any discrete random variable $U$, any probability distribution $q(u)$, and any classical channel $q(x|u)$, the corresponding input-output cq-states \[\sigma_{UR}=\sum_{x\in\set{X}}\sum_{u\in\set{U}}q(u)q(x|u)|u\>\<u|_U\otimes\sigma^x_R\qquad \textrm{and} \qquad \tau_{US}=\sum_{x\in\set{X}}\sum_{u\in\set{U}}q(u)q(x|u)|u\>\<u|_U\otimes\tau^x_S\] satisfy
	\[
	\Hmin(U|R)_{\sigma}\le \Hmin(U|S)_{\tau},
	\]
	namely, $\pg(U|R)_{\sigma}\ge \pg(U|S)_{\tau}$;
	\item $\mE$ is degradable into $\mE'$, i.e., there exists a CPTP map $\mT:\op{L}(\sH_R)\to\op{L}(\sH_S)$ such that
	\[
	\tau^x_S=\mT(\sigma^x_R),\qquad x\in\set{X}.
	\]
\end{enumerate}
\end{corollary}

\subsection{The classical case}\label{sec:fully-class}

When both the cq-channels have commuting output, we can state the result in purely classical terms as follows:

\begin{corollary}\label{coro:classical}
	Given two classical noisy channels $p(y|x)$ and $p'(z|x)$, the following are equivalent:
	\begin{enumerate}[label=\roman*)]
		\item $p$ is degradable into $p'$;
		\item for any discrete random variable $U$, any probability distribution $q(u)$, and any channel $q(x|u)$, the joint probability distributions $q(u)q(x|u)p(y|x)$ and $q(u)q(x|u)p'(z|x)$ satisfy
		\[
		\Hmin(U|Y)\le \Hmin(U|Z).
		\]
	\end{enumerate}
\end{corollary}

In other words, by replacing $H$ with $\Hmin$ in Definition~\ref{def:less-noisy}, we obtain that the corresponding notion of `less noisy' is equivalent to the notion of `degradable.' On the other hand, we recall the fact that there exist less noisy channels that are not degradable~\cite{Korner1977}. For the reader's convenience,  we report below a self-contained proof of Corollary~\ref{coro:classical}, which does not rely on any previous result about statistical morphisms or quantum channels.

\begin{proof}[Proof of Corollary~\ref{coro:classical}]
Obviously, point~(i) implies point~(ii). Conversely, let us assume point~(ii). This means that, for any joint probability distribution $q(x,u)$,
\[
\max_d\sum_{x,y,u}q(x,u)p(y|x)d(u|y)\ge \max_{d'}\sum_{x,z,u}q(x,u)p'(z|x)d'(u|z),
\]
where $d(u|y)$ and $d'(u|z)$ denote the guessing strategies, namely, discrete noisy channels $d:Y\to \hat U$ and $d':Z\to \hat U$, which the receiver can optimize in order to maximize the probability of correct guessing $\operatorname{Pr}\{U=\hat U \}$.

Choose now $U$ with $\set{U}=\set{Z}$, and label its states by $z'$. Also, fix the strategy $d'(z'|z)=\delta_{z',z}$. Then, for any $q(x,z')$, there exists $d(z'|y)$ such that
\[
\begin{split}
&\sum_{x,z'}q(x,z')\left(\sum_{z}p'(z|x)d'(z|z')-\sum_yp(y|x)d(z'|y)\right)\\
=\,&\sum_{x,z'}q(x,z')\left(p'(z'|x)-\sum_yp(y|x)d(z'|y)\right)\\
\le\,&0.
\end{split}
\]
Equivalently,
\[
\max_q\min_d\left\{\sum_{x,z'}q(x,z')\left(p'(z'|x)-\sum_yp(y|x)d(z'|y)\right)\right\}\le 0.
\]
By the minimax theorem (for our case, see Lemma~4.13 in Ref.~\cite{liese_statistical_2008}) we can exchange the order of the two optimizations, so that:
\[
\min_d\max_q\left\{\sum_{x,z'}q(x,z')\left(p'(z'|x)-\sum_yp(y|x)d(z'|y)\right)\right\}\le 0.
\]
Denoting by $\Delta(x,z')$ the difference $p'(z'|x)-\sum_yp(y|x)d(z'|y)$, we notice that, since $\sum_{x,z'}\Delta(x,z')=0$, we necessarily have that $\max_{x,z'}\Delta(x,z')\ge 0$, otherwise $\sum_{x,z'}\Delta(x,z')<0$. Consequently,
\[
\min_d\max_q\left\{\sum_{x,z'}q(x,z')\left(p'(z'|x)-\sum_yp(y|x)d(z'|y)\right)\right\}=\min_d\max_{x,z'}\Delta(x,z'),
\]
i.e., the maximum is achieved by concentrating the probability distribution $q(x,z')$ on one largest entry. Then, for what we said, we know that
\[
\min_d\max_{x,z'}\Delta(x,z')=0,
\]
implying the existence of a channel $d(z'|y)$ such that
\[
\sum_yp(y|x)d(z'|y)=p'(z'|x),\qquad\forall x,z',
\]
namely, $p$ is degradable into $p'$ as claimed.
\end{proof}

The main advantage of the above proof, with respect to the one used in the general case, is that it can be easily generalized to the approximate case, namely, when there exists $\epsilon\ge 0$ such that, for any random variable $U$ and any joint probability distribution $q(x,u)$,
\[
\pg(U|Y)\ge\pg(U|Z)-\epsilon.
\]
This case is studied in Appendix~\ref{sec:app-approximate}.

\section{Second extension result: the fully quantum case}\label{sec:quantum}

\begin{lemma}\label{lemma:general-channels}
	For a given CPTP map $\mN:\op{L}(\sH_Q)\to\op{L}(\sH_R)$ and a given linear map $\mL:\op{L}(\sH_R)\to\op{L}(\sH_S)$, let $\sH_{S'}\cong\sH_S$, and assume that $(\id_{S'}\otimes\mL)$ is a statistical morphism of $(\id_{S'}\otimes\mN)$. Then there exists a CPTP map $\mT:\op{L}(\sH_R)\to\op{L}(\sH_S)$ such that
	\[
\mT\circ\mN=\mL\circ\mN.
	\]
\end{lemma}

\begin{proof}
	Being $(\id_{S'}\otimes\mL)$ a statistical morphism of $(\id_{S'}\otimes\mN)$, we know, from Definition~\ref{def:morphisms}, that, for any POVM $\{\bar P^x_{S'S}:x\in\set{X} \}$ on $\sH_{S'}\otimes\sH_S\cong\sH_S^{\otimes 2}$, there exists a POVM $\{P^x_{S'R}:x\in\set{X} \}$ on $\sH_{S'}\otimes\sH_R$ such that
	\[
	\Tr\Big[\{\omega_{S'}\otimes(\mL\circ\mN)(\rho_Q)\}\ \bar P^x_{S'S}\Big]=\Tr\Big[\{\omega_{S'}\otimes(\mN)(\rho_Q)\}\  P^x_{S'R}\Big],
	\]
	for all $x\in\set{X}$, all $\omega_{S'}\in\op{S}(\sH_{S'})$, and all $\rho_Q\in\op{S}(\sH_Q)$. By linearity, this implies that
	\begin{equation}\label{eq:comp-telep}
\begin{split}
	&\Tr_{S'S}\Big[\{\Phi^+_{S''S'}\otimes(\mL\circ\mN)(\rho_Q)\}\  \{\openone_{S''}\otimes \bar P^x_{S'S}\}\Big]\\=&\Tr_{S'R}\Big[\{\Phi^+_{S''S'}\otimes(\mN)(\rho_Q)\}\  \{\openone_{S''}\otimes P^x_{S'R}\}\Big],
\end{split}
	\end{equation}
	for all $x\in\set{X}$ and all $\rho_Q\in\op{S}(\sH_Q)$,
	where $\Phi^+_{S''S'}=d_S\sum_{i,j=1}^{d_S}|i_{S''}\>|i_{S'}\>\<j_{S''}|\<j_{S'}|$ is the maximally entangled state in $\op{S}(\sH_{S''}\otimes\sH_{S'})\cong\op{S}(\sH_S^{\otimes 2})$.
	
	The protocol of generalized teleportation~\cite{general-teleport} implies the existence of a POVM $\{B^x_{S''S'}:x\in\set{X}\}$ and unitary operators $\{U^x_{S''\to S}:x\in\set{X} \}$ such that
	\[
	(\mL\circ\mN)(\rho_Q)=\sum_{x\in\set{X}}U^x_{S''\to S}
	\Tr_{S'S}\Big[\{\Phi^+_{S''S'}\otimes(\mL\circ\mN)(\rho_Q)\}\  \{\openone_{S''}\otimes B^x_{S'S}\}\Big](U^x_{S''\to S})^\dag,
	\]
	for all $\rho_Q\in\op{S}(\sH_Q)$. Then, Eq.~(\ref{eq:comp-telep}) implies the existence of a POVM $\{P^x_{S'R}:x\in\set{X} \}$ on $\sH_{S'}\otimes\sH_R$ such that
	\[
		(\mL\circ\mN)(\rho_Q)=\sum_{x\in\set{X}}U^x_{S''\to S}
		\Tr_{S'R}\Big[\{\Phi^+_{S''S'}\otimes(\mN)(\rho_Q)\}\  \{\openone_{S''}\otimes P^x_{S'R}\}\Big](U^x_{S''\to S})^\dag,
	\]
	for all $\rho_Q\in\op{S}(\sH_Q)$. 
	The statement is proved by defining the map $\mT:\op{L}(\sH_R)\to\op{L}(\sH_S)$ as
	\[
	\mT(Z_R)\defeq \sum_{x\in\set{X}}U^x_{S''\to S}
	\Tr_{S'R}\Big[\{\Phi^+_{S''S'}\otimes Z_R\}\  \{\openone_{S''}\otimes P^x_{S'R}\}\Big](U^x_{S''\to S})^\dag,
	\]
	and noticing that, being a sort of `noisy teleportation,' $\mT$  is indeed a CPTP map, as claimed.
\end{proof}

In fact, following an argument in Ref.~\cite{buscemi2012comparison}, it is not difficult to show that the assumption in Lemma~\ref{lemma:general-channels} can be somewhat weakened as follows: instead of assuming that $(\id_{S'}\otimes\mL)$ is a statistical morphism of $(\id_{S'}\otimes\mN)$, one can assume that $(\id_{S'}\otimes\mL)$ is a statistical morphism of $(\mD_{S'}\otimes\mN)$, where $\mD:\op{L}(\sH_{S'})\to\op{L}(\sH_{S'})$ is some invertible CPTP map, in the sense that $\mD(\op{L}(\sH_{S'}))=\op{L}(\sH_{S'})$. (For example, a channel $\mD(\rho)=p\rho+(1-p)\openone/d$ is invertible as long as $p>0$.)

As an immediate consequence of Lemma~\ref{lemma:general-channels} above and Proposition~\ref{prop:morphism}, we obtain the following:

\begin{corollary}
	Let $\mN:\op{L}(\sH_Q)\to\op{L}(\sH_R)$ and $\mN':\op{L}(\sH_Q)\to\op{L}(\sH_S)$ be two CPTP maps. Let $\sH_{S'}$ be an auxiliary Hilbert space such that $\sH_{S'}\cong\sH_S$. The following are equivalent:
	\begin{enumerate}[label=\roman*)]
		\item for any discrete random variable $U$, any probability distribution $q(u)$, and any cq-channel $\mE=\{\rho^u_{S'Q}:u\in\set{U} \}$, the corresponding input-output cq-states \[\sigma_{US'R}=\sum_{u\in\set{U}}q(u)|u\>\<u|_U\otimes(\id_{S'}\otimes\mN_Q)(\rho^u_{S'Q})\] and \[\tau_{US'S}=\sum_{u\in\set{U}}q(u)|u\>\<u|_U\otimes(\id_{S'}\otimes\mN'_Q)(\rho^u_{S'Q})\] satisfy
		\[
		\Hmin(U|S'R)_\sigma\le		\Hmin(U|S'S)_\tau,
		\]
		namely, $\pg(U|S'R)_\sigma\ge		\pg(U|S'S)_\tau$;
		\item $\mN$ is degradable into $\mN'$, i.e., there exists a CPTP map $\mT:\op{L}(\sH_R)\to\op{L}(\sH_S)$ such that
		\[
		\mN'=\mT\circ\mN.
		\]
	\end{enumerate}
\end{corollary}

In the case of cq-channels, we have the following:

\begin{corollary}
	Consider two cq-channels $\mE:\set{X}\to\op{S}(\sH_R)$ and $\mE':\set{X}\to\op{S}(\sH_S)$, with $\mE=\{\sigma_R^x:x\in\set
	X \}$ and $\mE'=\{\tau_S^x:x\in\set
	X \}$. Introduce an auxiliary Hilbert space $\sH_{S'}\cong\sH_S$ and let $\mE'':\set{Y}\to\op{S}(\sH_{S'})$ be a cq-channel, with $\mE''=\{\omega^y_{S'}:y\in\set{Y} \}$, such that $\operatorname{span}\{\omega^y_{S'}:y\in\set{Y}\}=\op{L}(\sH_{S'})$. Then, the following are equivalent:
	\begin{enumerate}[label=\roman*)]
		\item for any discrete random variable $U$, any probability distribution $q(u)$, and any classical channel $q(y,x|u)$, the corresponding input-output cq-states \[\sigma_{US'R}=\sum_{y\in\set{Y}}\sum_{x\in\set{X}}\sum_{u\in\set{U}}q(u)q(y,x|u)|u\>\<u|_U\otimes\omega^y_{S'}\otimes\sigma^x_R\] and \[\tau_{US'S}=\sum_{y\in\set{Y}}\sum_{x\in\set{X}}\sum_{u\in\set{U}}q(u)q(x|u)|u\>\<u|_U\otimes\omega^y_{S'}\otimes\tau^x_S\] satisfy
		\[
		\Hmin(U|S'R)_{\sigma}\le \Hmin(U|S'S)_{\tau},
		\]
		namely, $\pg(U|S'R)_{\sigma}\ge \pg(U|S'S)_{\tau}$;
		\item $\mE$ is degradable into $\mE'$, i.e., there exists a CPTP map $\mT:\op{L}(\sH_R)\to\op{L}(\sH_S)$ such that
		\[
		\tau^x_S=\mT(\sigma^x_R),\qquad x\in\set{X}.
		\]
	\end{enumerate}
\end{corollary}

\section{Conclusions}\label{sec:concls}

In this work we described a connection between the theory of statistical comparison and the comparison of noisy channels, independent from that of Ref.~\cite{Raginsky2011}. In particular, we showed how Definitions~\ref{def:degradable} and~\ref{def:less-noisy} become completely equivalent if $H$ is replaced by $\Hmin$ in Definition~\ref{def:less-noisy}.

The result proved here can be seen as a \textit{converse} to the data-processing inequality: the monotonic decrease of information (as measured here by $\Hmin$ or, equivalently, by $\pg$) is not only necessary but also \textit{sufficient} for the existence of a post-processing map (a trace preserving statistical morphism, in general, but we saw how additional assumptions can lead to the existence of a CPTP post-processing).

As we already mentioned somewhere else~\cite{buscemi2012all,buscemi_complete_2014,buscemi_equivalence_2014,buscemi_forthcoming_1}, we believe that this approach, based on the theory of statistical comparison, can play an important role in understanding the peculiarity of memoryless processes as the information-theoretic counterpart of adiabatic processes in thermodynamics.

\section*{Acknowledgments}

The author  acknowledges financial support from the JSPS KAKENHI, No. 26247016.

\appendix

\section{Classical case: approximate version}\label{sec:app-approximate}

Assuming
\[
\pg(U|Y)\ge\pg(U|Z)-\epsilon,\qquad\epsilon\ge0,
\]
the proof of Corollary~\ref{coro:classical} carries through unaltered, until one shows that
\begin{equation}\label{eq:class-opt}
\min_d\max_{x,z'}\Delta(x,z')\le\epsilon.
\end{equation}
To proceed from here, consider now the following quantity:
\begin{equation}\label{eq:var-dist}
\max_x\sum_{z'}|\Delta(x,z')|.
\end{equation}
The above quantity is the induced $l_1$-norm distance~\footnote{
	It holds that (see, e.g., Example 5.6.4 in Ref.~\cite{horn-johnson}) \[\max_x\sum_{z'}|\Delta(x,z')|=\max_{x}\sum_{z'}\left|p'(z'|x)-\sum_y
	p(y|x)d(z'|y)\right|\defeq \MN{p'-dp}_1,\] where $\MN{A}_1\defeq\max_{v:\N{v}_1=1}\N{Av}_1$ is the \textit{variational norm}. The quantity in Eq.~(\ref{eq:var-dist}) measures how well one can statistically distinguish $p'(z|x)$ from $\sum_yp(y|x)d(y|z)$.
}
between the channel $p'(z'|x)$ and the degraded channel $\sum_y
p(y|x)d(z'|y)$.
Since, for all $x$, $\sum_{z'}\Delta(x,z')=0$, we have that
\[
\sum_{z'}|\Delta(x,z')|=2\sum_{z':\Delta(x,z')\ge0}\Delta(x,z'),\qquad\forall x\in\set{X},
\]
which implies that, for the strategy $d$ achieving the left-hand side of Eq.~(\ref{eq:class-opt}),
\[
\begin{split}
\sum_{z'}|\Delta(x,z')|\le&\,2|\set{X}|\max_{x,z'}\Delta(x,z')\\
\le&\,2|\set{X}|\epsilon,\qquad\forall x\in\set{X}.
\end{split}
\]
In particular,
\[
\max_{x}\sum_{z'}\left|p'(z'|x)-\sum_y
p(y|x)d(z'|y)\right|\le2|\set{X}|\epsilon.
\]
We summarize this finding in a separate corollary:

\begin{corollary}
	Given two classical noisy channels, $p(y|x)$ and $p'(z|x)$, and $\epsilon\ge0$, suppose that, for any discrete random variable $U$, any probability distribution $q(u)$, and any channel $q(x|u)$, the joint probability distributions $q(u)q(x|u)p(y|x)$ and $q(u)q(x|u)p'(z|x)$ satisfy
		\[
		2^{-\Hmin(U|Y)}\ge 2^{-\Hmin(U|Z)}-\epsilon.
		\]		
		Then, there exists a degrading channel $d(z|y)$ such that
		\[
		\max_x\sum_z\left|p'(z|x)-\sum_yp(y|x)d(z|y)\right|\le2|\set{X}|\epsilon.
		\]
\end{corollary}

\end{document}